\documentclass[11pt]{article}

\usepackage[utf8]{inputenc}
\usepackage{amsmath, amsthm, amssymb}
\usepackage{xcolor}
\usepackage{algorithm2e}
\usepackage{graphicx}
\usepackage{fullpage}
\usepackage{tikz}
\usepackage{amsthm}
\newtheorem*{theorem*}{Theorem}

\newtheorem{theorem}{Theorem}[section]
\newtheorem{lemma}[theorem]{Lemma}
\newtheorem{proposition}[theorem]{Proposition}

\newtheorem{definition}[theorem]{Definition}

\newtheorem{question}{Question}

\newcommand{\E}[0]{\ensuremath{\mathbb{E}}}
\newcommand{\opt}{\ensuremath{\textsc{opt}}}
\newcommand{\OPT}{\ensuremath{\textsc{OPT}}}

\newcommand{\mvc}[1]{\tau\left(#1 \right)}
\newcommand{\MVC}{\text{MVC}}
\newcommand{\out}{\ensuremath{\Pi}}

\usepackage[numbers]{natbib}
\usepackage{hyperref}
\hypersetup{
    colorlinks=true,
    citecolor=blue,
    linkcolor=blue
}

\title{One-way Communication Complexity of Minimum Vertex Cover in General Graphs}
\author{
Mahsa Derakhshan, Andisheh Ghasemi, Rajmohan Rajaraman\\
Northeastern University, Boston, USA
}

\date{}

\begin{document}

\maketitle

\begin{abstract}
We study the communication complexity of the Minimum Vertex Cover (MVC) problem on general graphs within the \(k\)-party one-way communication model. Edges of an arbitrary \(n\)-vertex graph are distributed among \(k\) parties. The objective is for the parties to collectively find a small vertex cover of the graph while adhering to a communication protocol where each party sequentially sends a message to the next until the last party outputs a valid vertex cover of the whole graph. We are particularly interested in the trade-off between the size of the messages sent and the approximation ratio of the output solution.

It is straightforward to see that any constant approximation protocol for MVC requires communicating \(\Omega(n)\) bits. Additionally, there exists a trivial 2-approximation protocol where the parties collectively find a maximal matching of the graph greedily and return the subset of vertices matched. This raises a natural question: \textit{What is the best approximation ratio achievable using optimal communication of \(O(n)\)?} We design a protocol with an approximation ratio of \((2-2^{-k+1}+\epsilon)\) and \(O(n)\) communication for any desirably small constant \(\epsilon>0\), which is strictly better than 2 for any constant number of parties. Moreover, we show that achieving an approximation ratio smaller than \(3/2\) for the two-party case requires \(n^{1 + \Omega(1/\lg\lg n)}\) communication, thereby establishing the tightness of our protocol for two parties.

A notable aspect of our protocol is that no edges are communicated between the parties. Instead, for any \(1 \leq i < k\), the \(i\)-th party only communicates a constant number of vertex covers for all edges assigned to the first $i$ parties. An interesting consequence is that the communication cost of our protocol is \(O(n)\) bits, as opposed to the typical \(\Omega(n\log n)\) bits required for many graph problems, such as maximum matching, where protocols commonly involve communicating edges.

\end{abstract}

\section{Introduction}\label{sec2}

We study the communication complexity of the Minimum Vertex Cover (MVC) problem for general graphs within the $k$-party one-way communication model. In this model, the edges of an arbitrary $n$-vertex graph \( G = (V, E) \) are distributed among \( k \) parties, labeled from 1 to \( k \). The objective is for the parties to collectively find a small vertex cover of the graph \( G \) while adhering to the following communication protocol: Each party, in sequence, sends a message to the next party, starting from the first one, and ultimately, the last party outputs a valid vertex cover of \( G \). In this work, we focus particularly on the trade-off between the size of the messages sent and the approximation ratio of the output solution.

Extensive literature addresses various problems within this communication model, such as graph connectivity, set cover, minimum cut, and maximum matching. (See e.g., \cite{DBLP:conf/soda/GoelKK12, DBLP:conf/soda/Kapralov13, DBLP:conf/approx/SunW15, DBLP:conf/soda/AssadiK18, DBLP:conf/innovations/GhoshS24, DBLP:conf/icalp/AzarmehrB23}). The model, first introduced by Yao~\cite{DBLP:conf/stoc/Yao79} in 1979 also has close relations with streaming algorithms, which has further contributed to the significant attention it has received. However, to the best of our knowledge, this is the first work to consider the minimum vertex cover problem for general graphs in this communication model.

It is straightforward to see that a message of size \(\Omega(n)\) is necessary to achieve any constant approximation ratio for the Minimum Vertex Cover (MVC) problem; for example, if the entire graph is given to the first party and the MVC has size $n/2$. Moreover, with slight adjustments to a lower bound for the  maximum bipartite matching problem~\cite{DBLP:conf/soda/GoelKK12}, we prove that to achieve any approximation ratio smaller than $3/2$ in the two-party case, \(\omega(n)\) communication is necessary. On the positive side, there exists a trivial 2-approximation protocol where the parties collectively find a maximal matching of the graph greedily and return the subset of vertices matched in the maximal matching. Therefore, a natural question arises:

\begin{question}
    Is it possible to achieve better than a $2$-approximation with the optimal communication complexity of \(O(n)\)? 
\end{question}

We note that this question was, in fact, open even with $n^{2-\Omega(1)}$ communication. In this work, we answer the question affirmatively for any constant number of parties. We design a protocol with the approximation ratio being a function of the number of parties. This protocol is optimal when there are only two parties. The approximation ratio and communication complexity of our protocol are as follows. 

\newcommand{\maintheorem}[0]{For any $k \geq 2$ and any desirably small $\epsilon > 0$, there exists a randomized MVC protocol in the $k$-party one-way communication model with an expected approximation ratio of $(2 - 2^{-k+1} + \epsilon)$, in which each party communicates a message of size \(O_{k, \epsilon}(n)\)\footnote{We use the notation \(O_{k, \epsilon}(\cdot)\) to indicate that the hidden constant may depend on the parameters $k$ and $\epsilon$.}. This approximation ratio is tight for $k=2$ up to a factor of $1+\epsilon$.}

\begin{theorem} \label{thm:upper-mvc}
\maintheorem{}
\end{theorem}

An interesting aspect of our protocol is that no edges are communicated between the parties. Instead, for any \(1 \leq i < k\), the \(i\)-th party communicates only a constant number of vertex covers (dependent on \(\epsilon\) and \(k\)) of all the edges assigned to the first \(i\) parties. As a result, the communication cost of our protocol is \(O(n)\) bits, in contrast to the typical \(\Omega(n\log n)\) bits required for approximating many graph problems~\cite{DBLP:conf/approx/SunW15}, such as maximum matching, where protocols often involve communicating edges.

While the problem studied in this work is  information-theoretical rather than computational, it is worth noting that our protocol is not polynomial-time. This is expected, as achieving any approximation ratio better than $2$ for the minimum vertex cover problem is known to be computationally hard under the unique games conjecture~\cite{DBLP:journals/jcss/KhotR08}. Additionally, our protocol can be considered non-explicit because we provide an existential proof for parts of our protocol rather than explicitly constructing it in full (due to the use of von Neumann's Minimax Theorem). However, the protocol can, in principle, be identified through a brute-force search in doubly exponential time.

\paragraph*{Relation to streaming algorithms.} A key motivation behind studying the communication complexity of one-way protocols for graph problems is due to its relations with the streaming model~\cite{feigenbaum2005graph}. In the single-pass streaming model, the edges of an $n$-vertex graph arrive sequentially in a stream in an arbitrary order. A streaming algorithm needs to construct a solution using a limited memory smaller than the input size.  Any lower bounds established for the communication complexity in the one-way model directly imply lower bounds for the memory requirements in the streaming model. This, indeed, has been the standard approach for proving streaming lower bounds. Similarly, any protocol developed within the communication framework can provide insights into designing streaming algorithms.

A long-standing, central open question in streaming literature is whether it is possible to achieve better than a 0.5-approximate matching in \( \Tilde{O}(n) \) space with only a single pass. Even estimating the size of the maximum matching (the dual of minimum vertex cover) within a factor greater than 0.5 remains unresolved in this setting. The same holds true for surpassing a 2-approximation for Minimum Vertex Cover (MVC) in bipartite graphs, as well as the seemingly more challenging problem of approximating MVC on general graphs.

Our work highlights an important point for those aiming to prove the impossibility of better-than-2 approximation for streaming MVC within \( O(n) \) space: the standard approach of using the one-way communication model with a constant number of parties does not suffice, as we present a protocol that achieves better-than-2 approximation when \( k \) is a constant. On a positive note, we hope that our protocol will inspire the development of space-efficient streaming algorithms.

\subsection{Our Techniques}
In this section, we provide a brief overview of our protocol. A more detailed explanation of the protocol for two parties is presented in Section~\ref{sec:warmup}, and the multi-party protocol is formally introduced and analyzed in Section~\ref{sec:k-party}.  A lower bound for the two-party case is presented in Section~\ref{sec:lower bound}, and the proof of main theorem put together in Section~\ref{sec:main}.  

In our protocol, each party \(i\), for \(1 \leq i < k\), communicates a set of vertex covers of the edges assigned to parties \(1\) through \(i\). This information is clearly sufficient for the last party to obtain a valid vertex cover of the entire graph. Interestingly, it turns out that communicating only a constant number\footnote{This constant is independent of $n$, the number of vertices, but depends on $k$, the number of parties, and $\epsilon$.} of vertex covers provides enough information to find a sufficiently small vertex cover of the whole graph. To construct this message, each party \(i\) must solve several instances of the following problem: Let \(S\) be one of the vertex covers communicated to party \(i\), and let  \(\mathcal{D}_B\) be a distribution on the edges assigned to the future parties (\(i+1\) to \(k\)). Given this information, party \(i\) aims to add a subset of vertices to \(S\) to cover her edges while ensuring that this does not significantly increase the final approximation ratio.

Let $c_v$ denote the probability of the vertex belonging to the optimal solution given the distribution \(\mathcal{D}_B\). Player $i$ must choose which vertices to add to $S$. To make this decision, she assigns weights to the vertices, where higher weights show a poorer fit for inclusion in the solution. The player then computes the minimum-weighted vertex cover based on these assigned weights.
We aim to design the weight assignments so that vertices with higher weights are less likely to be selected. Additionally, we aim to provide earlier players with less flexibility in their choices, as they have access to less information. To achieve this, we define the weight function as:
\begin{equation}
    w_v= 1- (2-\beta_{k-i})c_v,
\end{equation}
where $(2-\beta_{k-i})$ is the approximation ratio achievable in a protocol with $k-i$ parties (the number of remaining parties). This means that if  $i$ is the second to last party (i.e., $k-i=1$), then she sets $w_v=1-c_v$ since a single party protocol finds the optimal solution ($\beta_{1} = 1$). Finally, analyzing the approximation ratio boils down to upper-bounding the weight of this minimum weight vertex cover which we discuss further in Sections~\ref{sec:ojirjgioer} and \ref{sec:technical}. Once this upper bound is established, we use an inductive proof to show that party \(i\) can add a subset of vertices to \(S\) ensuring
\begin{align}
    \frac{\E\big[|\text{solution outputted by the final party}|\big]}{\E\big[|\text{optimal solution}|\big]} \leq 2 - 2^{-k+1},
\end{align}
    
where the expectation is taken over \(\mathcal{D}_B\). Ideally, we would prefer this to be an instance-wise approximation ratio, where the upper bound holds for the expectation of the ratio rather than for the ratio of expectations. In such a case, a simple application of von Neumann's Minimax Theorem would suffice to establish the existence of our desired protocol. Nevertheless, this upper bound, combined with a simple trick (also used by Assadi and Behnezhad~\cite{DBLP:conf/approx/AssadiB21}), can still give us the same approximation ratio with a small additive loss. 


The key idea is that if the Minimum Vertex Cover of all instances in the support of the distribution $\mathcal{D}_B$
  have nearly the same size, then the per-instance approximation ratio and the ratio of expectations will also be approximately equal.
To use this, each party discretizes the range of possible optimal solution sizes into a constant number of intervals. Each party then solves a separate instance of the problem under the assumption that the true optimal solution size falls within a particular interval. For each of these size estimates, the party is facing a distinct problem instance.
The number of size estimates (or "guesses") each party makes depends on their position among the $k$ parties. Later parties in the protocol are allowed more error, meaning they make a larger number of guesses to refine their estimate.
Ultimately, the total number of vertex covers that the last party receives is proportional to the product of the guesses made by the first $k-1$ parties.
See Figure~\ref{fig:kparty} for a depiction of this protocol.

\subsection{Related Work}

    In recent works, several advancements in the communication complexity of various graph-related problems have been made. Assadi et al. \cite{assadi2016tight} established optimal lower bounds for approximating the Set Cover problem in the two-party communication model, which also corresponds with a streaming algorithm, thereby addressing both streaming and communication complexities. Additionally, Abboud et al. \cite{abboud2021smaller} investigated the same communication model for exact answers, which included back-and-forth communication, while Dark et al. \cite{dark2020optimal} modified this model to include deletions. The model has also been studied by Assadi et al. \cite{assadi2017randomized} in a simultaneous framework under random partitions. Moreover, Naidu and Shah \cite{naidu2022space} as well as Chitnis et al. \cite{chitnis2016kernelization} contributed relevant insights in the dynamic data stream model, particularly concerning the vertex cover problem.

Most related to our work is the rich literature on the communication complexity of the maximum matching problem. (See e.g., \cite{DBLP:conf/soda/GoelKK12,DBLP:conf/icalp/Bernstein20,DBLP:conf/soda/AssadiB19,DBLP:journals/talg/LeeS20}) Within the $\Tilde{O}(n)$ one-way communication regime, Goel, Kapralov, and Khanna have established a protocol for bipartite matching that achieves a tight approximation ratio of $2/3$ in the two-party case \cite{DBLP:conf/soda/GoelKK12}. For general graphs, the same tight approximation ratio is achieved by Assadi and Bernstein \cite{DBLP:conf/soda/AssadiB19}. In scenarios involving more than two parties ($k>2$), the work of Behnezhad and Khanna~\cite{DBLP:conf/soda/BehnezhadK22} implies the existence of protocols with an approximation ratio of $0.6$ and $0.53$, respectively, for three and four parties, and $0.5+\Omega_{1/k}(1)$ for any $k>4$. Furthermore, Singla and Lee~\cite{DBLP:journals/talg/LeeS20} provide  $0.5+2^{-O(k)}$ approximation protocols for the more restricted model where each party has to irrevocably add some of their edges to the matching instead of communicating an arbitrary message of size $\Tilde{O}(n)$.  

Given the duality between maximum matching and minimum vertex cover (MVC) on bipartite graphs, these results can be used to approximate the size of bipartite MVC. However, simply having an approximate matching does not suffice to construct an approximate MVC. Although it would not be surprising if some of these works also have implications for finding an approximate MVC for bipartite graphs, to the best of our knowledge, there is no explicit study of MVC within the model discussed in this paper, even for bipartite graphs. General graphs, however, which are the main focus of this work, are a completely different story, as the MVC can be up to two times the maximum matching.  Therefore, the ideas developed for approximating matching matching are not of much use here. 

Another line of work related to this paper is the study of the stochastic minimum vertex cover problem~\cite{DBLP:conf/soda/BehnezhadBD22, DBLP:conf/stoc/DerakhshanDH23, DBLP:conf/innovations/Derakhshan25}.  In this problem, we are given a graph $G=(V, E)$ and an existence probability for each edge $e \in E$. Edges of $G$ are realized (or exist) independently with these probabilities, forming the realized subgraph $\mathcal{G}$. The existence of an edge in $\mathcal{G}$ can only be verified using edge queries. The goal of this problem is to find a near-optimal vertex cover of $\mathcal{G}$ using a small number of queries.  Most related to our work is the 1.5 approximation algorithm designed by Derakhshan et al~\cite{DBLP:conf/stoc/DerakhshanDH23}. Their algorithm first commits a subset of vertices $S$ to the final solution, then queries any edge in $G$ which is not covered by $S$, and adds an MVC of the realized edges to the solution. This effectively results in a solution which covers all edges of $\mathcal{G}$. The core of their algorithm is their choice of set $S$.

The stochastic vertex cover algorithm of~\cite{DBLP:conf/stoc/DerakhshanDH23} can be used to attack the special case of the two-party communication model as follows. Let us assume that Alice (the first party) in addition to her own graph also has a distribution over Bob's graph (the second party). In other words, she has a distribution over the whole graph with her edges having an existence probability of one. (Unlike the stochastic model, existence of edges are not independent here.) We observe that in this case if Alice follows the algorithm of~\cite{DBLP:conf/stoc/DerakhshanDH23} and only communicates subset $S$ to  Bob, they will be able to find a vertex cover whose expected size is at most 1.5 times the expected size of the optimal solution. Moreover, we show that the assumption regarding the knowledge of distribution can be lifted via techniques from~\cite{DBLP:conf/approx/AssadiB21} allowing for a 1.5 approximation in the 2-party communication model. However, for the multi-party case, which is the main focus of this work, while the technical insights from this work are helpful, it does not imply any approximation ratio better than 2.  To address this challenge, we extend the argument of~~\cite{DBLP:conf/stoc/DerakhshanDH23} to bound the cost of a suitably weighted vertex cover in the 2-party case and leverage that to further generalize to the $k$-party case, beating the threshold of 2 for the approximation ratio.


\subsection{Preliminaries}
 In this work, we consider the one-way 
$k$-party communication model, focusing on randomized protocols where each party utilizes private random bits. The primary interest in this model is information-theoretical, with parties assumed to be computationally unbounded. The communication cost is measured by the worst-case length of the messages exchanged between any two parties. For standard definitions and further details, we refer to the textbook by Kushilevitz and Nisan~\cite{DBLP:books/daglib/0011756}.

We have a base graph $G = (V, E)$ which is distributed among $k$ parties in a way that each party $i$ has the graph $G_i = (V, E_i)$ where $\bigcup_{l= 1}^{k} E_l = E$. The objective of the parties is to collectively find a small vertex cover of the graph \( G \) while adhering to the following communication protocol: Starting from the first party, each party $1 \leq i < k$, in turn, sends a message $M_i$ and ultimately, the last party outputs a valid vertex cover of \( G \).

\noindent \textbf{Notation.} We define $\MVC(.)$ as a function that, for any given graph, returns its minimum vertex cover. The size of the minimum vertex cover for an input graph is denoted by $\mvc{.}$, and $\OPT$ represents the minimum vertex cover of the input  graph $G$, implying that $|\OPT| = \mvc{G}$ is the optimal solution size.

\section{An Overview of the Two-Party Protocol}
\label{sec:warmup}

As a warm-up, we consider a two-party version of the problem with
Alice and Bob as the first and second parties, respectively. Let \(
G_A=(V, E_A) \) and \( G_B=(V, E_B) \) represent the subgraphs given
to Alice and Bob. In this section, we provide an overview of our
protocol for two parties with $O(n)$ communication for a
constant $\epsilon > 0$ that can be made arbitrarily small. We later prove this is a \(3/2 +
\epsilon\)-approximate protocol. 

Consider a simplified version of the problem where Alice, in addition
to \( G_A \), also knows the size of the minimum vertex cover of the
whole graph \( G = G_A \cup G_B \). We will later discuss how to lift
this assumption. For this case, the protocol we consider is simple:
Alice communicates a carefully constructed vertex cover \( X \) of her
subgraph (not necessarily the minimum vertex cover) to Bob.  This
construction is randomized, and we argue that it is possible to
pick \( X \) that will lead to an expected approximation ratio of
\(3/2\).


Let us analyze the approximation ratio achieved by a fixed \( X\). The final solution picked by Bob will be the union of \( X \) and
a minimum vertex cover of \( G_B[V \backslash X] \). Hence, the
approximation ratio will be:
\begin{equation} \label{eq:jiiref}
    \frac{\E[|X|+\mvc{G_B[V \backslash X]}]}{\mvc{G}}
\end{equation}

\noindent \textbf{A Two-Player Game.}  
We can view the problem faced by Alice as a zero-sum two-player game
between her and an adversary, both of whom know \( G_A \) and the size
of the optimal solution represented by \(\opt\). Alice's strategy is
to select a subset \( X \subseteq V \) of vertices that covers \( G_A
\). The adversary's strategy is to select \( G_B \) such that
\(\tau(G_B \cup G_A)=\opt\). We let the adversary's utility be the
approximation ratio defined in Equation~\ref{eq:jiiref} since that is
what Alice wants to minimize. Note that a mixed strategy of Alice in
this game (a distribution over vertex covers of \( G_A \)) is
equivalent to a randomized algorithm for picking \( X \) in the
original problem. Therefore, we are interested in proving that Alice
has a mixed strategy for selecting \( X \) such that no pure strategy
of the adversary obtains utility larger than 3/2 against it. The Minimax Theorem due to von Neumann~\cite{vonNeumann:1928:TGG}
implies the following about this game:

\begin{proposition}\label{prop:kjpiwqe9ji}
In the described game, if for any mixed strategy chosen by the
adversary, there exists a pure strategy \( X \) for Alice such that
the adversary's expected utility is at most \( 3/2 \), then there
exists a mixed strategy for Alice that ensures the expected utility of
the adversary is at most \( 3/2 \) for any strategy the adversary
employs.
\end{proposition}

Due to Proposition~\ref{prop:kjpiwqe9ji},
the problem is reduced to proving that given any distribution \(
\mathcal{D}_B \) over \( G_B \), there exists a vertex cover \( X \)
of \( G_A \) with an instance-wise approximation ratio of at most
3/2. Given distribution $\mathcal{D}_B$, let \( c_v \) denote the
probability of a vertex being in the optimal solution of $G_A \cup
G_B$, assuming $G_B$ is drawn from $\mathcal{D}_B$.  That is,
\[ c_v := \Pr_{G_B \sim \mathcal{D}_B}[v \in \MVC(G_A \cup G_B)]. \]

After calculating these probabilities (which is possible without computational limitations), Alice constructs a
vertex-weighted subgraph with the weight of any vertex \( v \in V \)
being \( w_v= 1-c_v \). She then lets \( X \) be the minimum weight
vertex cover of \( G_A \).  These steps are formalized in
Algorithm~\ref{alg:329}.

\RestyleAlgo{boxruled}
\LinesNumbered
\begin{algorithm}[ht]
  \caption{Minimum Weighted Vertex Cover 2-Party Algorithm\label{alg:329}}
  Let \( c_v = \Pr[v \in \MVC(G_A \cup G_B)] \) given distribution \( \mathcal{D}_B \).\\
  Let \( w_v = 1 - c_v \) for all \( v \in V \).\\
  Alice finds a minimum weighted vertex cover \( X \) of \( G_A \).\\
  Alice sends \( X \) to Bob.\\
  Bob finds a minimum vertex cover \( M \) of \( G_B[V \backslash X] \).\\
  Bob outputs \( X \cup M \).
\end{algorithm}

We now provide some intuition behind Algorithm~\ref{alg:329}. Let \(
W(X) \) be the sum of the weight of the vertices in \( X \). We claim
that \( W(X) \) is the expected difference between \(\opt\) and the
size of the solution outputted by the algorithm. For any vertex in \(
X\), its weight is essentially the difference between its
contribution to \( X \) and its contribution to the optimal
solution. In other words, this weight is the {\em cost} of including
that vertex in \( X \).
We observe that since Bob calculates the
exact minimum vertex cover (MVC) for the edges that \( X \) does not
cover, the total weight \( W(X) \) acts as an upper bound on the
expected difference between the optimal solution and the solution
found by the algorithm.  This is exactly what we want to minimize,
which is why Alice sets weights as \( w_v= 1-c_v \) and takes a
minimum weight vertex cover of \( G_A \).

Following this observation, to prove the desired approximation ratio, we need to show an upper bound of 
\begin{equation}\label{eq:weightik}
    W(X)\leq \E_{G_B \sim \mathcal{D}_B}[\mvc{G}]/2. 
\end{equation} 
Proving this turns out to be a technically challenging problem, which
we tackle in Section~\ref{sec:technical} (in more generality)\footnote{Inequality~\ref{eq:weightik}  follows from Lemma~\ref{claim:vertex cover weight} with parameter $\beta=1$.}.  Our proof uses the
probabilistic method and shows that the expected weight of a
particular randomized vertex cover of $G_A$ is at most half the
optimal vertex cover cost of $G$.  To construct this randomized vertex
cover, we follow a similar approach as the stochastic vertex cover
algorithm of~\cite{DBLP:conf/stoc/DerakhshanDH23}.  In particular, we organize the vertices into
three groups based on the $c_v$ values and a suitable threshold $t
> 1/2$: $c_v \in (0, 1 - t]$, $c_v \in (1 - t, t]$, $(t,
    1]$.  We include all vertices from the third group.
    The vertices in the first group can be omitted from a vertex
      cover as all their edges in $G_A$ are to vertices in the third group. Finally, we include
      those vertices from the middle group that are selected in a
      minimum vertex cover of a graph $G_A \cup G_B$ drawn randomly
      based on the distribution $\mathcal{D}_B$.
      We extend the
      argument of~\cite{DBLP:conf/stoc/DerakhshanDH23} to analyze the weight of the above
      randomized vertex cover, which we then leverage to both
      establish the inequality~\ref{eq:weightik} as well as crucially generalize
      to the case of $k$ players.

Using inequality~\ref{eq:weightik}, we then prove for the 2-player case that\begin{equation}
     \frac{|X|+ \E_{G_B \sim \mathcal{D}_B}[\mvc{G_B[V \backslash X]}]}{\E_{G_B \sim \mathcal{D}_B}[\mvc{G}]} \leq 3/2. \end{equation} 
Here, $|X|+ \E_{G_B \sim \mathcal{D}_B}[\mvc{G_B[V \backslash X]}]$ is the expected size of the output since Bob returns the union of $X$ and a minimum vertex cover of the remaining graph. Observe that by definition of weights, for any vertex we have $c_v+w_v=1$.  As a result we can write 
\begin{align*}
    |X|+ \E[\mvc{G_B[V \backslash X]}] \leq \sum_{v\in X} (c_v+ w_v) + \sum_{v\notin X} c_v = \E[\mvc{G}] + \sum_{v\in X} w_v  \stackrel{\eqref{eq:weightik}}{\leq} 3\E[\mvc{G}]/2.  
\end{align*}
This, however, is not an instance-wise approximation ratio.  By
invoking Yao's minimax principle and using the assumption that the
size of the solution is fixed, we are able to ensure that there is a
randomized protocol that achieves an instance-wise approximation ratio
of 3/2 under the fixed size assumption.

\noindent\textbf{Lifting the Knowledge of Size.}
So far, we have discussed our protocol, assuming that Alice knows the
size of the optimal solution. To lift this assumption, we use a
standard approach which is also used by Assadi and Behnezhad~\cite{DBLP:conf/approx/AssadiB21}. Given a constant \(
\epsilon \in (0,1) \), we create \( \lceil \log_{1+\epsilon}2
\rceil\) instances of the problem. In the \( i \)-th instance, Alice
guesses the size of the optimal solution to be in the range
\[ ((1+\epsilon)^{i-1} \mvc{G_A}, (1+\epsilon)^i \mvc{G_A}). \]
For each of her guesses, she communicates a different vertex cover of
\( G_A \). It is easy to show that if one of these guesses is correct,
the subset Alice communicates for that guess allows Bob to find a
\(3/2+\epsilon\) approximation ratio. These guesses together cover the
case of \( \mvc{G} < 2\mvc{G_A} \). To address the case of \( \mvc{G}
\geq 2\mvc{G_A} \), Alice also communicates an MVC of her graph. In
this scenario, the size of the output would be at most
\[ \mvc{G} + \mvc{G_A} \leq  3\mvc{G}/2. \]
Putting the pieces together, Alice can guarantee an expected approximation ratio of \(3/2+\epsilon\) by communicating a message of size \( n \lceil \log_{1+\epsilon}2\rceil \).

\section{The $k$-Party Protocol}
\label{sec:k-party} 
In this section, we present the general $k$-party protocol, which
builds on the sketch of the $2$-party protocol presented in
Section~\ref{sec:warmup}. Our protocol simulates a series of two-party protocols between any party $i$ and a second party encapsulating all the remaining parties $i+1$ to $k$.

However, it is important to note that a naive approach of composing arbitrary
$3/2$-approximate two-party protocols may not yield an effective $k$-party protocol.  For instance, consider the graph presented in Figure~\ref{fig:example_graph} consisting of sets $A$, $B$, $C$, and $D$, each containing $n/4$ vertices. These sets are connected by two bipartite matchings, $M_1$ between $A$ and $B$, and $M_2$ between $C$ and $D$, and a complete graph $K$ over $B \cup C$.  For this graph, the optimal vertex cover is $B \cup C$ of size $n/2$.  Suppose the first party receives $M_1$, the second party receives $M_2$ and the third party receives $K$.  A $3/2$-approximate two-party protocol may select $M_1$ for the first party and $M_2$ for the second party, leading to an overall solution of size $n$, which would be $2$-approximate.  We address this challenge by carefully selecting vertices, prioritized by the $c_v$ values, and setting weights so that the optimal solution for the remaining graph becomes progressively smaller as the protocol proceeds with the parties.  

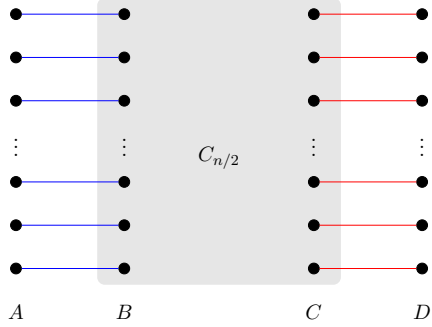
\begin{figure}
    \centering
    \begin{minipage}{0.5\textwidth} 
        \centering
        \scalebox{0.72}{
    \begin{tikzpicture}
        \def\dx{2} 
        \def\dy{0.8} 
        \def\gap{1.5} 
        \def\widegap{3.5} 

        \fill[gray!20, rounded corners] (-0.5, 0.3) rectangle (\widegap + 0.5, -4 * \dy - \gap - 0.3);
        
        \node at (\widegap / 2, -2 * \dy - \gap / 2 - 0.3) {$C_{n/2}$};

        \node[draw, circle, fill=black, inner sep=2pt] (farleft1) at (-\dx, 0) {};
        \node[draw, circle, fill=black, inner sep=2pt] (farleft2) at (-\dx, -\dy) {};
        \node[draw, circle, fill=black, inner sep=2pt] (farleft3) at (-\dx, -2 * \dy) {};
        \node at (-\dx, -2 * \dy - \gap / 2) {$\vdots$}; 
        \node[draw, circle, fill=black, inner sep=2pt] (farleft4) at (-\dx, -2 * \dy - \gap) {};
        \node[draw, circle, fill=black, inner sep=2pt] (farleft5) at (-\dx, -3 * \dy - \gap) {};
        \node[draw, circle, fill=black, inner sep=2pt] (farleft6) at (-\dx, -4 * \dy - \gap) {};
        \node at (-\dx, -4 * \dy - \gap - 1 + 0.2) {$A$};

        \node[draw, circle, fill=black, inner sep=2pt] (left1) at (0, 0) {};
        \node[draw, circle, fill=black, inner sep=2pt] (left2) at (0, -\dy) {};
        \node[draw, circle, fill=black, inner sep=2pt] (left3) at (0, -2 * \dy) {};
        \node at (0, -2 * \dy - \gap / 2) {$\vdots$}; 
        \node[draw, circle, fill=black, inner sep=2pt] (left4) at (0, -2 * \dy - \gap) {};
        \node[draw, circle, fill=black, inner sep=2pt] (left5) at (0, -3 * \dy - \gap) {};
        \node[draw, circle, fill=black, inner sep=2pt] (left6) at (0, -4 * \dy - \gap) {};
        \node at (0, -4 * \dy - \gap - 1 + 0.2) {$B$};

        \node[draw, circle, fill=black, inner sep=2pt] (right1) at (\widegap, 0) {};
        \node[draw, circle, fill=black, inner sep=2pt] (right2) at (\widegap, -\dy) {};
        \node[draw, circle, fill=black, inner sep=2pt] (right3) at (\widegap, -2 * \dy) {};
        \node at (\widegap, -2 * \dy - \gap / 2) {$\vdots$}; 
        \node[draw, circle, fill=black, inner sep=2pt] (right4) at (\widegap, -2 * \dy - \gap) {};
        \node[draw, circle, fill=black, inner sep=2pt] (right5) at (\widegap, -3 * \dy - \gap) {};
        \node[draw, circle, fill=black, inner sep=2pt] (right6) at (\widegap, -4 * \dy - \gap) {};
        \node at (\widegap, -4 * \dy - \gap - 1 + 0.2) {$C$};

        \node[draw, circle, fill=black, inner sep=2pt] (farright1) at (\widegap + \dx, 0) {};
        \node[draw, circle, fill=black, inner sep=2pt] (farright2) at (\widegap + \dx, -\dy) {};
        \node[draw, circle, fill=black, inner sep=2pt] (farright3) at (\widegap + \dx, -2 * \dy) {};
        \node at (\widegap + \dx, -2 * \dy - \gap / 2) {$\vdots$}; 
        \node[draw, circle, fill=black, inner sep=2pt] (farright4) at (\widegap + \dx, -2 * \dy - \gap) {};
        \node[draw, circle, fill=black, inner sep=2pt] (farright5) at (\widegap + \dx, -3 * \dy - \gap) {};
        \node[draw, circle, fill=black, inner sep=2pt] (farright6) at (\widegap + \dx, -4 * \dy - \gap) {};
        \node at (\widegap + \dx, -4 * \dy - \gap - 1 + 0.2) {$D$};

        \foreach \i in {1,2,3,4,5,6} {
            \draw[blue] (farleft\i) -- (left\i);
        }

        \foreach \i in {1,2,3,4,5,6} {
            \draw[red] (right\i) -- (farright\i);
        }

    \end{tikzpicture}
     }
    \end{minipage}%
    \begin{minipage}{0.45\textwidth} 
        \centering
          \caption{The blue matching $M_1$ is given to the first party, the red matching $M_2$ to the second party, and a complete graph (non-bipartite) among vertices in the box to the third party.}
            \label{fig:example_graph}
    \end{minipage}
\end{figure}

We give an overview of our $k$-party protocol.  Following
the framework of the 2-party protocol, the first party makes several
guesses about the size of $\opt$, with the number of guesses depending
only on $k$ and parameter $\epsilon > 0$, which can be made
arbitrarily small.  For each guess (which is an interval) the first party constructs a vertex cover for her graph and communicates
all these vertex covers to the next party. Figure~\ref{fig:kparty} depicts our protocol.

For any $1 < i < k$, the $i$-th party receives a message $M$, which is
a set of sets of vertices. Each element $S \in M$ is a vertex cover of
$\cup_{1 \leq l < i} G_l$. For each $S \in M$, party $i$ constructs a
new instance of the problem. She assumes $S$ will be in the final
vertex cover and lets $G_i' = G_i[V\backslash S]$ be the subgraph of
$G_i$ not covered by $S$. At this point, party $i$ faces a similar
problem to the first party in a $(k-i)$-party setting with the
assigned subgraph being $G_i'$.  Therefore, she starts by making a
number of guesses about $\mvc{G[V\backslash S]}$.  We use $b_i$ to refer to the number of guesses party $i$ needs to make and later discuss its exact value.  For each guess, she
picks a vertex cover of $G_i'$ and takes its union with $S$ as a
vertex cover of $\cup_{1 \leq j \leq i} G_j$. Repeating this for all
elements of $M$ results in constructing \[|M| \times [\text{number of
    guesses on } \mvc{G[V\backslash S]} \text{ for any } S \in M]\]
vertex covers of $\cup_{1 \leq j \leq i} G_j$. She then communicates
all of these to the next party. Since the number of guesses is only a
function of $k$ and $\epsilon$, the size of the message is $O_{k,
  \epsilon}(n)$.

Finally, the last party, after receiving a message $M$, finds the
smallest vertex cover of her graph that includes at least one subset
of vertices $S \in M$. This ensures that the final output is a valid
vertex cover for the entire graph.

Two points remain to be addressed about the protocol. Consider party
$i$ and the message $M$ she receives. First, for any $S \in M$, we
need to describe how the party forms guesses about the size of the
minimum vertex cover of $G[V \setminus S]$. Party $i$ makes $b_i =
\log_{1 + \epsilon}2^{k-i}+1$ guesses about $\mvc{G[V \setminus
    S]}$. For any $1<l<b_i$, the $l$-th guess is: \[ (1 +
\epsilon)^{l-1} \mvc{G'_i} \leq \mvc{G[V \setminus S]} < (1 +
\epsilon)^{l} \mvc{G'_i}.\] 
There is an additional guess
$\mvc{G[V \setminus S]} \geq (1 + \epsilon)^{b_i} \mvc{G'_i}$ to cover
the remaining possibilities.

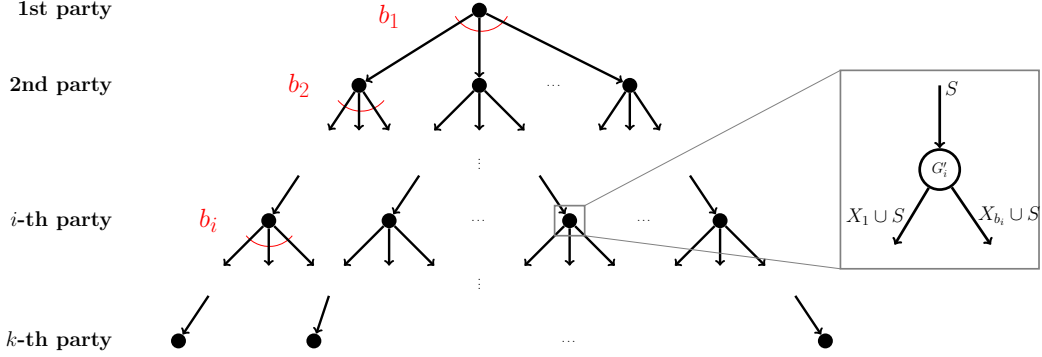
\begin{figure}[h]
    \centering

\scalebox{0.4}{ 

\begin{tikzpicture}
    \node[circle, fill=black, inner sep=5pt] (A) at (0,0) {};
    \node[left][scale=1.4] at (-12,0) {\Large \textbf{1st party}};  

    \draw[red, thick] (-0.8,-0.2) arc (210:330:1);
    \node[red] at (-3,-0.2) {\Huge $b_1$};
    
    \def\y{-2.5}
    \node[circle, fill=black, inner sep=5pt] (B1) at (-4,\y) {};
    \node[circle, fill=black, inner sep=5pt] (B2) at (0,\y) {};
    \node[circle, fill=black, inner sep=5pt] (B3) at (5,\y) {};
    \node at (2.5,\y) {\dots};
    \node[left][scale= 1.4] at (-12,\y) {\Large \textbf{2nd party}};  

    \draw[red, thick] (-4.7,\y-0.5) arc (220:320:1);
    \node[red] at (-6,\y) {\Huge $b_2$};
    
    \draw[->, line width=0.8mm] (A) -- (B1);
    \draw[->, line width=0.8mm] (A) -- (B2);
    \draw[->, line width=0.8mm] (A) -- (B3);

    \draw[->, line width=0.8mm] (B1) -- ++(-1,-1.5);
    \draw[->, line width=0.8mm] (B1) -- ++(0,-1.5);
    \draw[->, line width=0.8mm] (B1) -- ++(1,-1.5);

    \draw[->, line width=0.8mm] (B2) -- ++(-1.5,-1.5);
    \draw[->, line width=0.8mm] (B2) -- ++(0,-1.5);
    \draw[->, line width=0.8mm] (B2) -- ++(1.5,-1.5);

    \draw[->, line width=0.8mm] (B3) -- ++(-1,-1.5);
    \draw[->, line width=0.8mm] (B3) -- ++(0,-1.5);
    \draw[->, line width=0.8mm] (B3) -- ++(1,-1.5);

    \node at (0, \y-2.5) {\vdots};

    \def\yy{-7}
    \node[circle, fill=black, inner sep=5pt] (C1) at (-7,\yy) {};
    \node[circle, fill=black, inner sep=5pt] (C2) at (-3,\yy) {};
    \node[circle, fill=black, inner sep=5pt] (C3) at (3,\yy) {};
    \node[circle, fill=black, inner sep=5pt] (C4) at (8,\yy) {};
    \node[left][scale=1.4] at (-12,\yy) {\Large \textbf{$i$-th party}};  

       \draw[red, thick] (-7.75,\yy-0.5) arc (220:320:1);
    \node[red] at (-9,\yy) {\Huge $b_i$};

    \draw[->, line width=0.8mm] ++(-6,\yy+1.5) -- (C1);
    \draw[->, line width=0.8mm] ++(-2,\yy+1.5) -- (C2);
    \draw[->, line width=0.8mm] ++(2,\yy+1.5) -- (C3);
    \draw[->, line width=0.8mm] ++(7,\yy+1.5) -- (C4);

    \node at (0,\yy) {\dots};
    \node at (5.5,\yy) {\dots};

    \draw[->, line width=0.8mm] (C1) -- ++(-1.5,-1.5);
    \draw[->, line width=0.8mm] (C1) -- ++(0,-1.5);
    \draw[->, line width=0.8mm] (C1) -- ++(1.5,-1.5);

    \draw[->, line width=0.8mm] (C2) -- ++(-1.5,-1.5);
    \draw[->, line width=0.8mm] (C2) -- ++(0,-1.5);
    \draw[->, line width=0.8mm] (C2) -- ++(1.5,-1.5);

    \draw[->, line width=0.8mm] (C3) -- ++(-1.5,-1.5);
    \draw[->, line width=0.8mm] (C3) -- ++(0,-1.5);
    \draw[->, line width=0.8mm] (C3) -- ++(1.5,-1.5);

    \draw[->, line width=0.8mm] (C4) -- ++(-1.5,-1.5);
    \draw[->, line width=0.8mm] (C4) -- ++(0,-1.5);
    \draw[->, line width=0.8mm] (C4) -- ++(1.5,-1.5);

    \draw[gray, ultra thick] (2.5,\yy-0.5) rectangle (3.5,\yy+0.5);

    \node at (0, \yy-2) {\vdots};

    \def\gap{-11}
    \node[circle, fill=black, inner sep=5pt] (F1) at (-10,\gap) {};
    \node[circle, fill=black, inner sep=5pt] (F2) at (-5.5,\gap) {};
    \node[circle, fill=black, inner sep=5pt] (F3) at (11.5,\gap) {};
    \node[left][scale=1.4] at (-12,\gap) {\Large \textbf{$k$-th party}};  

    \draw[->, line width=0.8mm] (-9,\gap+1.5) -- (F1);
    \draw[->, line width=0.8mm] (-5,\gap+1.5) -- (F2);
    \node at (3,\gap) {\dots};
    \draw[->, line width=0.8mm] (10.5,\gap+1.5) -- (F3);

    \def\xOffset{12}  
    \def\yOffset{\yy+5}  
    \def\squareSize{6.6}  
    \draw[gray, ultra thick] (\xOffset, \yOffset) rectangle (\xOffset+\squareSize, \yOffset-\squareSize);

    \node[circle, draw=black, line width=0.9mm, inner sep=6pt] (G) at (\xOffset+\squareSize/2, \yOffset-\squareSize/2) {\Large $G'_i$};

    \draw[->, line width=0.9mm] (\xOffset+\squareSize/2, \yOffset-0.5) -- (G);
    
    \node [scale = 1.3] at (\xOffset+\squareSize/2 + 0.4, \yOffset-0.6) {\Large $S$};

    \draw[->, line width=0.9mm] (G) -- ++(-1.5,-2.5) node[midway, left][scale=1.3] {\Large $X_1 \cup S$};
    \draw[->, line width=0.9mm] (G) -- ++(1.7,-2.5) node[midway, right][scale= 1.3] {\Large $X_{b_i} \cup S$};

    \draw[gray, thick] (\xOffset, \yOffset) -- (3.5,\yy+0.5);  
    \draw[gray, thick] (\xOffset, \yOffset-\squareSize) -- (3.5,\yy-0.5); 

\end{tikzpicture}

}

       \caption{\small The tree of communications generated by parties using algorithm \ref{algk}. The message $S$ sent through each edge to the $i$-th layer is a vertex cover of the subgraphs given to the first $i-1$ parties. For any of these messages the $i$-th party receives, she constructs $b_i$ different vertex covers of $G'_i=G_i[V\setminus S]$ and sends their union with $S$ to the next party. Each edge basically represents a subproblem party $i$ needs to solve. Given $S$, she will make $b_i$ guesses about the size of $\mvc{G'_i}$ and for any of these guesses, needs to solve the following subproblem: Find a vertex cover $X$ of $G'_i$ such that committing the subset $S\cup X$ to the final solution results in a good approximation ratio conditioned on the specific guess about $\mvc{G'_i}$. 
    }
    \label{fig:kparty}
\end{figure}
    
The second point to address is the following: given a guess on the
optimal vertex cover, how should we construct a vertex cover of $G_i'
= G_i[V \setminus S]$ in order to obtain the desired approximation
ratio?  We discuss this in detail in Section~\ref{sec:ojirjgioer} via
the formulation of a two-player game that captures the decision
process for party $i$.  We then formalize the protocol in
Section~\ref{sec:protocol} and establish its communication complexity
and approximation ratio.

\subsection{A Two-Player Game}\label{sec:ojirjgioer}
In this section, we will discuss the sub-problems each party needs to solve in the design of our protocol. In Figure~\ref{fig:kparty}, the construction of a message sent through any edge represents one of these sub-problems. Consider any subset of vertices $S$ party $i$ receives from the previous party. (For the first party, this is an empty set.),  and let $V'= V\setminus S$. Assuming that $S$ will be part of the final solution and any guess the party has about the size of the optimal solution on the remaining graph $G'=G[V']$, she needs to send a message to the next party. The message will be the union of $S$ and a vertex cover $X$ of her remaining subgraph $G'_i=G_i[V']$. 

To simplify the problem, we can ignore $S$ since it is part of both the input message and the outgoing message. Therefore, in the simplified version, party $i$ needs to solve the following problem. The input consists of a subgraph  $G'_i = (V', E'_i)$ and a guess about the size of the remaining optimal solution in the form of $\mvc{G'}\in (O, (1+\epsilon)O)$  where $O$ is an integer number. The goal of the party is to find a vertex cover $X$ of  $G'_i$ in order to minimize the size of the final solution if subset $X$ is committed to the final solution knowing that the rest of the parties will implement a protocol with an expected approximation ratio of at most $2-2^{i-k}+5\epsilon$ on the remaining  subgraph $G'[V'\setminus X]$.

We can view these sub-problems as two-party problems between party $i$, which we will refer to as Alice, and a second party, Bob, who encapsulates parties $i+1$ to $k$ in the original problem. In this problem, Bob, instead of being able to find an exact MVC of the remaining graph, can only find a solution with an approximation ratio of at most $2-2^{i-k}+5\epsilon$. In this two-party problem, we use $G_A = (V', E_A)$ to refer to $G'$ and $G_B = (V', E_B)$ as the union of the subgraphs given to parties $i+1$ to $k$ in the original problem. Hence, Alice's goal here is to pick a randomized $X$ such that it results in an expected approximation of at most $ 2-2^{i-k-1}+5\epsilon$ for  any priori fixed $G_B$. 

We reformulate this two-party problem as a two-player game. The first player is Alice, and the second player is an adversary who will determine $G_B$. 

\begin{definition}[MVC Game] \label{def:Game} Given three parameters \(\beta \in (0,1)\), $\epsilon\in (0,1)$ and $O\in (0,n)$ we define a two-player zero-sum game between Alice and an adversary on a graph with vertex set \(V'\). Alice has a set of edges \(E_A\) between vertices \(V'\), which is known to both players. Alice's strategy is to select a subset \(X \subseteq V'\) of vertices that covers all edges in \(E_A\). The adversary's strategy is to select a set of edges \(E_B\) over \(V'\) satisfying \[
O\leq |\mathrm{MVC}(E_B \cup E_A)|\leq (1+\epsilon)O.\] The adversary's utility is defined as:
\begin{equation} \label{utility}
    U_B(X, E_B) = \frac{|X| + (2-\beta) \times |\mathrm{MVC}(E_B[V \setminus X])|}{|\mathrm{MVC}(E_B \cup E_A)|}.
\end{equation}
The notation \(E_B[V \setminus X]\) represents the edges in \(E_B\) that are induced by the vertices in the set \(V \setminus X\), meaning it includes only those edges in \(E_B\) where both endpoints are in \(V \setminus X\). 
\end{definition}

We begin by proving that for any given distribution over adversary's
strategies, there exists a deterministic strategy for Alice that
achieves a payoff of at least $(2-\beta/2)(1+\epsilon)$.

\begin{lemma}\label{lemma:pure}
    In the MVC game presented in Definition~\ref{def:Game}, given any
    randomized (mixed) strategy $E_B$ of the adversary, Alice has a
    deterministic (pure) strategy $X$ such that \[
    \E[U_B(X, E_B)]\leq
    (2-\beta/2)(1+\epsilon).\]
\end{lemma}
\begin{proof}
For any vertex $v$ let us define $c_v=\Pr[v\in \MVC(E_B\cup E_A)]$.   Alice chooses a minimum weight vertex cover $X$ of $E_A$, where the weight of a vertex is defined as $w_e= 1-(2-\beta)c_v$. We will prove that this choice of $X$ satisfies the statement of the lemma. We have 
\begin{align}
 \nonumber  \E\big[|X|+ (2-\beta)\big|\MVC(E_B[V \setminus X])\big|\big] & \leq \left( \sum_{v\in X} 1 \right)+   (2-\beta)\sum_{v\in V \setminus X} c_v  \\  \nonumber &= \left (\sum_{v\in X} (1- (2-\beta) c_v +  (2-\beta) c_v) \right) + (2-\beta)\sum_{v\in V \setminus X} c_v \\& = \left ( \sum_{v\in X} (1- (2-\beta) c_v)  \right )+ (2-\beta) \sum_{v\in V} c_v \label{eq:jegjregko}
\end{align}
In the first inequality, we use the fact that 
\begin{equation}\label{eq:jefj}
    \E[|\MVC(E_B[V \setminus X])|] \leq \sum_{v\in V \setminus X} c_v.
\end{equation} 
This is due to the definition of $c_v$s, which implies there is a vertex cover of $E_B\cup E_A$, containing any vertex $v\in V \setminus X$ w.p. $c_v$. Since edges in  $E_B[V \setminus X]$ can only be covered by vertices in $V \setminus X$, this implies that there is also a vertex cover of $E_B[V \setminus X]$ which contains each vertex w.p. $c_v$ hence by  linearity of expectation, we get \eqref{eq:jefj}.

 A technical part of our proof is to show $\sum_{v\in X} (1-
(2-\beta) c_v)\leq \frac{\beta}{2}\sum_{v\in V} c_v$ which we defer to
Lemma~\ref{claim:vertex cover weight} in
Section~\ref{sec:technical}. Combining this with \eqref{eq:jegjregko}
gives us
\begin{align*}
   \E[|X|+ (2-\beta)|\MVC(E_B[V \setminus X])|] &\leq   \sum_{v\in X} (1- (2-\beta) c_v) + (2-\beta) \sum_{v\in V} c_v \\ &\leq \frac{\beta}{2}\sum_{v\in V} c_v + (2-\beta) \sum_{v\in V} c_v\\ &\leq (2-\beta/2) \sum_{v\in V} c_v \\&\leq  (2-\beta/2) \E[\MVC(E_A \cup E_B)] \leq (2-\beta/2)(1+\epsilon)O.
\end{align*}
Hence, we get
\vspace{2 mm}
\begin{align*}
    \E\left[\frac{|X|+ (2-\beta)|\MVC(E_B[V \setminus X])|}{|\MVC(E_A \cup E_A)|}\right] &\leq \frac{\E\left[|X|+ (2-\beta)|\MVC(E_B[V \setminus X])|\right]}{O} \\&\leq 
    \frac{(2-\beta/2)(1+\epsilon)O}{O}\leq (2-\beta/2)(1+\epsilon),
\end{align*}
completing the proof of the lemma.
\end{proof}

Combining Lemma~\ref{lemma:pure} with von Neumann's Minimax Theorem yields
the existence of a randomized strategy for Alice that achieves an
expected payoff of at least $(2-\beta/2)(1+\epsilon)$.

\begin{lemma} \label{lemma:mixed}
In the MVC game presented in  Definition~\ref{def:Game}, Alice has a randomized strategy $X$ such that no strategy $E_B$ played by the adversary obtains utility larger than $(2-\beta/2)(1+\epsilon)$.
\end{lemma}

\begin{proof}
Given any parameter $\alpha > 0$, von Neumann's Minimax Theorem~\cite{vonNeumann:1928:TGG} (alternatively, Yao's minimax principle~\cite{DBLP:conf/focs/Yao77}) implies that in the described game, if for any mixed
strategy chosen by the adversary, there exists a pure strategy \( X \)
for Alice such that the adversary's expected utility is at most
$\alpha$, then there exists a mixed strategy for Alice that ensures
the expected utility of the adversary is at most $\alpha$ for any
strategy the adversary employs.  By Lemma~\ref{lemma:pure}, for any
mixed strategy of the adversary, there exists a pure strategy for
Alice that achieves an expected utility of at least
$(2-\beta/2)(1+\epsilon)$.  This implies the existence of a mixed
strategy for Alice with the same expected utility against an arbitrary
adversary.
\end{proof}

\subsection{The Protocol}
\label{sec:protocol} In this section, we provide a formal statement of our protocol in  Algorithm~\ref{algk}.
We then establish in Lemma~\ref{lemma:communication} that the total communication complexity is $O(n)$
for any constant $k$. Finally, we prove in Lemma~\ref{lemma:approximation} that it results in our desired approximation ratio. 
\RestyleAlgo{boxruled} \LinesNumbered
\begin{algorithm}[ht]
\caption{$k$-Party Algorithm for party $i$\label{algk}}
    \textbf{Input:} subgraph $G_i$ and $M_{i-1}$ (message from party $i-1$ and $M_0 = \{ \emptyset \}$.)\\
     Let $\epsilon\in (0, 1/5)$ be a given constant number.\\
    \For{$S \in M_{i-1}$}{
    $G_i' = G_i[V \backslash S]$ (Remove all the vertices in $S$ from the graph.)\\
        \If{$i < k$}{
        \For{$l= 1$ to $b_i$}{
        $O \gets (1 + \epsilon)^{l-1} \mvc{G_i'}$. \\
        $\beta \gets 2^{-k+i} - 5\epsilon$.\\
        Apply the strategy of Lemma \ref{lemma:mixed} with $O$ and $\beta$ to get $X_l$.\\
        Add $X_l \cup S$ to $M_i$.
        }}
        Add $\MVC(G_i') \cup S$ to $M_i$.
    }
    \eIf{$i = k$}{
    \text{\Return the set $\out_i$ with minimum size in $M_i$.}
    }
    {send $M_i$ to party $i+1$}
    
\end{algorithm}

As mentioned before, parts of our protocol are non-explicit. In particular, we only show the existence of suitable subsets $X_1, \dots, X_{b_i}$.  In our protocol, each party $i$ communicates a set of suitable vertex covers of the graph given to the first $i$ parties. We do not present an efficient procedure for computing these covers, which makes our protocol non-explicit.

\begin{lemma}
  \label{lemma:communication}
Given a graph $G = (V, E)$ and $k$ parties, if all the parties use algorithm \ref{algk} to communicate and find the vertex cover of $G$, the communication cost will be 
    $(k-1)! (\log_{1 + \epsilon}2)^{k-1} O(n)$.
\end{lemma}

\begin{proof}
    Since each party sends $b_i$ new subsets for each subset they receive, the number of subsets communicated between the last two parties is 
    \[b_1 b_2 ... b_{k-1} = \log_{1 + \epsilon}2^{k-1} \cdot \log_{1 + \epsilon} 2^{k-2} \cdots \log_{1 + \epsilon} 2^{1} = (k-1)! (\log_{1 + \epsilon} 2)^{k-1}\]
    Moreover, since each subset of vertices can be represented using an $n$-bit string, the communication cost is $(k-1)! (\log_{1 + \epsilon}2)^{k-1} O(n) = O_{\epsilon, k} (n)$.
\end{proof}

We next establish an upper bound on the approximation
ratio achieved by the protocol.

\begin{lemma}
\label{lemma:approximation}
  In the $k$-party minimum vertex cover problem, 
let $\pi_{k}(G)$ be the size of the output if all the parties use algorithm \ref{algk} on the base graph $G$.
     Then \[\E [\pi_k(G) ] \leq (2 - \frac{1}{2^{k-1}} + 5 \epsilon) \mvc{G}\]
\end{lemma}
\begin{proof}
    We will use induction on the number of parties.
    If $k=1$, then the algorithm will output the set with the minimum size in $M_1$. Since $M_0 = \{ \emptyset \}$, a set in $M_1$ will be $\mvc{G'= G[V \backslash \emptyset] = G}$ that is added on line 11. Therefore, the output is $\mvc{G}$ which gives us
    \[\E[\pi_{1}(G)] \leq \mvc{G} \leq (2 - \frac{1}{2^{1-1}} + 5 \epsilon) \mvc{G} = (1 + 5\epsilon)\mvc{G},\] proving base case for $k=1$.
     For the inductive step, assume that the lemma statement holds for $k-1$ parties and any base graph $G'$. That is 
    \begin{equation}
    \label{eq:inductive}
       \E[ \pi_{k-1}(G')] \leq (2 - \frac{1}{2^{k-2}} + 5 \epsilon) \mvc{G'}.
    \end{equation}
    In line 6 of the algorithm \ref{algk}, let $X_{1}, X_{2},..., X_{b_1}$ be the subsets generated by the first party using the strategy of Lemma \ref{lemma:mixed} with $O$ and $\beta$ of line 9, where $b_1 = \log_{1 + \epsilon} 2^{k-1}$.
    Let $j$ be the integer number satisfying $ (1 + \epsilon)^{j-1} \mvc{G_1} \leq \mvc{G} \leq (1 + \epsilon)^{j} \mvc{G_1}$. 

    \textbf{Case 1}: $j \leq b_1 $.
    Party $2$ upon receiving $X_{j}$, constructs a graph $G' = G[V \backslash X_{j}]$ since $X_{j}$ is committed to be in the final solution and we do not need to consider the edges of these vertices. We find a vertex cover on $G'$ where the partitions are $G'_l = G_l[V \backslash X_{j}]$ for  $2 \leq l \leq k$. Since the lemma statement holds for $k-1$ parties, applying algorithm \ref{algk} on a graph $G'$ gives us a vertex cover of $G'$ with the expected approximation ratio of \[(2 - \frac{1}{2^{(k-1)-1}} + 5 \epsilon).\]
    Since $\out_k$ covers all the edges in $G[V \backslash X_{j}]$, a vertex cover for $G$ is $\out_k \cup X_{j}$.

    As a result of this, the final solution will have size $\pi_{k}(G) = |X_{j}| + \pi_{k-1}(G')$ and results in the following inequality for the approximation ratio.
    \begin{equation}
    \label{eq:eq2}
        \frac{\E[\pi_{k}(G)]}{\mvc{G}} = \frac{\E[|X_{j}|] + \E[\pi_{k-1}(G')]}{\mvc{E_A \cup E_B}} \leq \frac{|X_{j}| + (2 - \frac{1}{2^{k-1}} + 5\epsilon) \mvc{G'}}{\mvc{E_A \cup E_B}}
    \end{equation}

    Now, we construct a version of the MVC game in which we have Alice as the first party and
    the adversary as the union of the next $k-1$ parties. Let $\beta = \frac{1}{2^{k-1}} - 5 \epsilon$, $O = (1 + \epsilon)^{j-1} \mvc{G_1}$ and $\epsilon = \frac{1}{5}$.
     As described in the game definition, Alice sends a vertex cover to Bob and aims to minimize the utility function \ref{eq:1234}. The adversary (the next $k-1$ parties) aims to choose edges $E_B$ in a way to maximize his utility. Let $G_B = (V, E_B = \bigcup_{l= 2}^{k} E_l)$. The adversary's strategy is to select a set of edges $E_B$ over $V$ s.t. $O \leq |\mvc{E_B \cup E_A}| \leq (1 + \epsilon) O$. 
     \begin{equation}
     \label{eq:1234}
         \E[U_B(X_{j} , E_B)] = \frac{|X_{j}| + (2 - \beta)\mvc{E_B[V \backslash X_{j}]}}{\mvc{E_B \cup E_A}} = \frac{|X_{j}| + (2 - \frac{1}{2^{k-2}} + 5\epsilon)\mvc{G'}}{\mvc{E_B \cup E_A}}
     \end{equation}

    Now using Lemma~\ref{lemma:mixed}, Alice has a strategy such that
    no strategy played by the adversary obtains utility larger than
    $(2 - \beta/2)(1 + \epsilon)$.
    \begin{equation}\label{eq:iurhwodi}
        \frac{\E[\pi_k(G)]}{\mvc{G}}  \leq \E [U_B(X_{j}, E_B)]  \leq (2 - \beta/2)(1 + \epsilon)
    \end{equation}
   
    By putting $\beta$ in Equation~\ref{eq:iurhwodi} we get
  
    \begin{align*}
        \E [U_B(X_{j}, E_B)] & \leq (2 - \beta/2)(1 + \epsilon) = (2 - \frac{1}{2^{k-1}} + 5/2 \epsilon)(1 + \epsilon) \\
        & = 2 - \frac{1}{2^{k-1}} + 5/2 \epsilon  + 2\epsilon - \frac{1}{2^{k-1}} \epsilon + 5/2 \epsilon^2 = 2 - \frac{1}{2^{k-1}} + \epsilon [9/2 - \frac{1}{2^{k-1}} + 5/2 \epsilon]. 
    \end{align*} 
    Since $\epsilon \leq \frac{1}{5}$, $9/2 -\frac{1}{2^{k-1}} + 5/2 \epsilon \leq 5.$   The following completes our inductive step for case 1.
    \[\frac{\pi_k(G)}{\mvc{G}} \leq (2 - \beta/2)(1 + \epsilon) \leq 2 - \frac{1}{2^{k-1}} + 5 \epsilon.\]

    \textbf{Case 2:} $j > b_i$. So, $2^{k-1} \mvc{G_1} < \mvc{G}$. Since $\mvc{G_1}$ is a subset in $M_{1}$, party 2 constructs $G' = G[V \backslash MVC(G_1)]$.  Since $\mvc{G'} \leq \mvc{G}$, by Equation~\ref{eq:inductive}, we get \[\E[\pi_{k-1}(G')] \leq (2 - \frac{1}{2^{k-2}} + 5\epsilon) \mvc{G'} \leq (2 - \frac{1}{2^{k-2}} + 5\epsilon) \mvc{G}.\]
    As explained in the previous case, the output is at least the summation of the minimum vertex cover of $G_1$ and the algorithm's output on the rest of the partitions.
    \[\E[\pi_{k}(G)] \leq \mvc{G_1} + \pi_{k-1}(G') \leq \frac{1}{2^{k-1}} \mvc{G} + (2 - \frac{1}{2^{k-2}} + 5\epsilon) \mvc{G} \leq (2 - \frac{1}{2^{k-1}} + 5 \epsilon ) \mvc{G}.\]
   This completes the proof for case 2 which completes the proof.
\end{proof}

The upper bound of Theorem~\ref{thm:upper-mvc} follows from
Lemmas~\ref{lemma:communication} and~\ref{lemma:approximation}.

\subsection{A Technical Lemma}
\label{sec:technical}
In this section, we establish the following lemma, which is used in
the proof of Lemma~\ref{lemma:pure}.
\begin{lemma}
  \label{claim:vertex cover weight}
$\sum_{v\in X} (1- (2-\beta) c_v)\leq \frac{\beta}{2}\sum_{v\in V} c_v.$
\end{lemma}
It helps to first establish the following claim that may be of independent interest.  
\begin{lemma}
  \label{lem:middle region}
Let $w: (0,1] \rightarrow \Re_{\ge 0}$ be a function with finite
  support\footnote{We define the support of $w$ to be the number of
    elements $x \in (0,1]$ for which $w(x) > 0$.}.  Then, we have
\[
\left(\forall x \in [1/2, 1]: \sum_{y \in [x,1]} w(y) \cdot y \ge \sum_{y \in [0,1-x]} w(y) \cdot y\right) \implies
\sum_{y \in (0,1]} w(y) \cdot y (1 - 2y) \le 0.
\]
\end{lemma}

\label{app:middle region}
\begin{proof}
  Our proof is by induction on the size of the support of $w$.  For
  the base case, we let the support be empty (so, $w(x) = 0$ for all
  $x$), in which case the claim holds trivially.  Fix integer $n > 1$.
  Suppose the claim of the lemma holds for all functions with
  support less than $n$.  Let $w$ be a function with support of
  size $n$ that satisfies the condition of the lemma.

  Since the support of $w$ is positive, $\sum_{y \in (0,1]}
    w(y) \cdot y > 0.$  By the condition of the lemma we have $\sum_{y \in
      \cap [1/2,1]} w(y) \cdot y \ge \sum_{y \in (0,1/2]} w(y) \cdot
    y.$  Therefore, $\sum_{y \in [1/2,1]} w(y) \cdot y > 0.$  It
    follows that there exists $y \ge 1/2$ with $w(y) > 0$.  Let
    $\alpha$ be the minimum $y \in [1/2, t]$ with $w(y) > 0$.

  Let $T$ denote $\{y \in [1-\alpha, 1/2): w(y) > 0\}$.  We consider
    two cases depending on whether $T$ is empty.  We first consider
    the case where $T$ is empty.  Let $w'$ be identical to $w$ except
    that $w'(\alpha) = 0$.  By our definition of $w'$, we have that
    for all $x \in (\alpha, 1]$
\begin{eqnarray*}
  \sum_{y \in [x,1]} w'(y) \cdot y = \sum_{y \in [x,1]} w(y) \cdot y \ge  
\sum_{y \in [0,1-x]} w(y) \cdot y =  
\sum_{y \in [0,1-x]} w'(y) \cdot y.
\end{eqnarray*}
Since $w'(y) = 0$ for all $y \in [1-\alpha,\alpha]$, we obtain that
the condition of the lemma holds for $w'$, whose support is strictly
less than that of $w$.  By the induction hypothesis, we have
$\sum_{y \in (0,1]} w'(y) \cdot y (1 - 2y) \le 0$,
which yields the following since $\alpha \ge 1/2$.
\[
\sum_{y \in (0,1]} w(y) \cdot y (1 - 2y) = w(\alpha) \alpha(1-2\alpha)
  + \sum_{y \in (0,1]} w'(y) \cdot y (1 - 2y) \le 0.
\]
We now consider the case where $T$ is nonempty.  We define a new
function $w_1: (0,1] \rightarrow \Re_{\ge 0}$.
  \[
  w_1(y) = \left\{
  \begin{array}{ll}
    0 & y \in (1 - \alpha, 1/2)\\
    \sum_{y \in [1-\alpha, 1/2)} w(y) y/(1-\alpha) & y = 1 - \alpha\\
    w(y) & \mbox{otherwise}
  \end{array}
  \right.
  \]
We observe that for any $x > \alpha$ we have \[\sum_{y \in [x,1]} w_1(y) y =
\sum_{y \in [x,1]} w(y) y \,\,\,\,\,\,\,\,\, \,\,\text{  and} \,\,\,\,\,\,\,\, \,\, \sum_{y \in (0,1-x]} w_1(y) y =
    \sum_{y \in [x,1]} w(y) y.\]  Therefore, the condition of the lemma
    holds for $w_1$ for $x > \alpha$.  For $x \le \alpha$, we have
\begin{eqnarray*}
  \sum_{y \in [x,1]} w_1(y) y & = & \sum_{y \in [1/2,1]} w_1(y) y\\
  & \ge & \sum_{y \in (0,1/2]} w(y) y\\
  & = & \sum_{y \in [0,1-\alpha)} w(y) y + \sum_{y \in [1-\alpha]} w(y) y\\
  & = & \sum_{y \in [0,1-\alpha)} w_1(y) y +  w_1(1-\alpha) (1 - \alpha)\\
  & = & \sum_{y \in (0,1/2]} w(y) y.
\end{eqnarray*}
Thus, $w_1$ satisfies the lemma condition.  We next define
$w_2$ to be identical to $w_1$ except that \[w_2(\alpha) = w_1(\alpha)
- \min\{w_1(1-\alpha), w_1(\alpha)\} \;\;\;\; \text{and} \;\;\;\; w_2(1-\alpha) =
w_1(1-\alpha) - \min\{w_1(1-\alpha), w_1(\alpha)\}.\]  Note that
$w_2(\alpha) = 0$ or $w_2(1 - \alpha) = 0$, so the support of $w_2$ is
strictly less than the support of $w_1$ and hence that of $w$.
Also, $w_2$ satisfies the lemma condition.  Therefore,
we derive
\begin{eqnarray*}    
 & &  \sum_{y \in (0,1]} w(y) \cdot y (1 - 2y)\\
 & = & \sum_{y \in (0,1]} w_1(y) \cdot y (1 - 2y) -
  w_1(1-\alpha) \cdot (1 - \alpha) (2\alpha -1) + 
  \sum_{y \in [1-\alpha,1/2)} w(y) \cdot y (1 - 2y)\\ 
  & = & \sum_{y \in (0,1]} w_1(y) \cdot y (1 - 2y) -
  \sum_{y \in (1-\alpha,1/2)} w(y) \cdot y (2\alpha -1) + 
  \sum_{y \in (1-\alpha,1/2)} w(y) \cdot y (1 - 2y)\\ 
  & \le & \sum_{y \in (0,1]} w_1(y) \cdot y (1 - 2y)
   + \sum_{y \in (0,1]} w_2(y) \cdot y (1 - 2y)\\
   & \le & 0,
\end{eqnarray*}
where the last step follows from the induction hypothesis.  This
completes the proof of the induction step and that of the lemma.
\end{proof}

\noindent\textbf{Proof of Lemma~\ref{claim:vertex cover weight}.}
    Consider a vertex cover $I := \{ v : c_v > t \} \cup \{MVC(G_i \sim \mathcal{G}) \cap v :
    c_v \in (1-t , t) \}$, where $MVC(G_i \sim \mathcal{G})$ is the
    minimum vertex cover of a graph $G_i$ drawn randomly according to the distribution
    $\mathcal{G}$ and $t$ is the smallest $t \in [0.5, 1]$ such that
    $\sum_{v : c_v > t} c_v \leq \sum_{v: c_v < 1-t} c_v.$
    Since $X$ is a minimum weighted vertex cover of $G^A$, we have that 
    \[\sum_{v \in X} w_v \leq \underset{G_i \sim \mathcal{G}}{\E}[\sum_{v \in I} w_v]\]
    Let $S_1 = \{v: c_v < 1-t\}$, $S_2 = \{v: 1-t \le c_v \le t\}$, and $S_3 = \{v: c_v > t\}$.  Then,
   \[\E \left[\sum_{(S_1 \& S_3) \cap I} w_v\right] = \sum_{S_3} w_v = \sum_{S_3} (1 - 2c_v + \beta c_v) \leq \sum_{S_3} \beta c_v \leq \frac{\beta}{2} \left(\sum_{S_3} c_v + \sum_{S_1} c_v\right) = \frac{\beta}{2} \sum_{S_1 \& S_3} c_v,\]
    where we used $c_v \ge 1/2$ for $v \in S_3$ and $\sum_{S_3} c_v \leq \sum_{S_1} c_v$ by definition of $t$.

    We now consider the set $S_2$.  We derive
    \begin{eqnarray*}
      \E \left[\sum_{S_2 \cap I} w_v\right] & = & \sum_{v \in S_2} c_v(1 - (2-\beta)c_v) = \sum_{v \in S_2} \left(c_v(1 - 2c_v) + \beta c_v^2\right)\\
      & = & \sum_{v \in S_2} \left(c_v(1 - 2c_v) - \frac{\beta}{2} c_v(1-2c_v) + \frac{\beta}{2} c_v\right)\\
      & = & \left(1 - \frac{\beta}{2}\right) \sum_{v \in S_2} c_v(1 - 2c_v) + \frac{\beta}{2} \sum_{v \in S_2} c_v
      \le \frac{\beta}{2} \sum_{v \in S_2} c_v,
    \end{eqnarray*}
    where the last inequality follows from Lemma~\ref{lem:middle
      region} (with $w(x)$ set to the number of vertices $v$ for which
    $c_v = x$, if $x \in S_2$, and 0 if $x \notin S_2$).

    This yields ${\E}[\sum_{v \in I} w_v] \le \frac{\beta}{2} \sum_{v} c_v$
    implying that $\sum_{v \in X} w_v \le \frac{\beta}{2} \sum_{v}
    c_v$.\qed

\newcommand{\eps}{\varepsilon}
\section{A Lower Bound for the Two-Party Case}
\label{sec:lower bound}
We show the following lower bound, which establishes the tightness of our protocol for the two-party case.  
\begin{lemma}
\label{lemma:lower-bound}
For any constant
$\eps > 0$, there exists a constant $c > 0$ such that any vertex cover
computed by a two party protocol with $n^{1 + c/\lg\lg n}$ communication complexity
has an approximation ratio of at least $3/2 - \eps$.
\end{lemma}
\begin{proof}
    
\label{app:lower bound}
Let $G$ be a Rusza-Szemeredi (bipartite) graph $(P, Q, E)$ formed by
$n$ vertices on each side and $k$ induced matchings $M_1 \ldots, M_k$,
where $k = n^{\Omega(1/\lg\lg n)}$ and $|M_i| = (1/2 - \eps_1)n$ for $1 \le i
\le k$, where $\eps_1 > 0$ is a constant that can be made arbitrarily
small.  Rusza-Szemeredi graphs with the preceding parameters have been
shown to exist~\cite{DBLP:conf/soda/GoelKK12}.  We generate a random
bipartite graph $G' = (P \cup P', Q \cup Q', E_1 \cup E_2)$, with a total of 
$(3+2\eps_1)n$ vertices, as follows:
\begin{enumerate}
\item
  $P$, $Q$ are as in $G$; $P'$ (resp., $Q'$) is a set of $(1/2 +
  \eps_1)n$ vertices disjoint from $P$ (resp., $Q$).
\item
  Select a subset $M_i'$ of $\eps_2 n$ edges uniformly at random from
  $M_i$, for $1 \le i \le k$, independently, where $\eps_2 > 0$ is a
  constant that can be set arbitrarily small.  Set $E_1 = \cup_{1 \le
    i \le k} M'_i$; we thus have $|E_1| = k \eps_2 n$.
\item
  Choose $r$ uniformly at random from $\{1, \ldots, k\}$.  Let $P_r$
  and $Q_r$ denote the vertices of $M_r$ in $P$ and $Q$, respectively.
  Let $M_P$ and $M_Q$ denote a perfect matching between $P'$ and $P
  \setminus P_r$ and between $Q'$ and $Q \setminus Q_r$,
  respectively.  Set $E_2$ to $M_P \cup M_Q$.
\end{enumerate}
In the two-party instance, Alice receives the bipartite graph $(P
\cup P', Q \cup Q', E_1)$ and Bob receives the bipartite graph
$(P \cup P', Q \cup Q', E_2)$.  The minimum vertex cover of $G'$
has size at most $\eps_2 n + 2(1/2 + \eps_1)n = (1 + \eps_2 +
2\eps_1)n$; for example, if $\widetilde{P}$ denotes the set of vertices of $M'_r$
in $P_r$, then $\widetilde{P} \cup (P \setminus P_r) \cup (Q \setminus Q_r)$ is a
vertex cover of this size.  We will argue below that with probability
at least $1 - o(1)$, any protocol with $O(n)$ communication produces a
vertex cover of size at least $(3/2 + 2\eps_2 - \eps_3)n$, where
$\eps_3 > 0$ is a constant that can be made arbitrarily small by
suitably setting other constants, if necessary.

The proof is via a counting argument.  Let ${\cal G}_A$ denote the
collection of all possible graphs that can be presented to Alice.  By
our construction above, we have
\[
|{\cal G}_A| \ge \binom{(1/2 - \eps_1)n}{\eps_2 n}^k \ge \left(\frac{(1/2 - \eps_1)n}{\eps_2 n}\right)^{\eps_2 n k}. 
\]
Suppose Alice sends a message with $s$ bits to Bob; let $\phi: {\cal
  G}_A \rightarrow \{0,1\}^s$ denote the mapping that Alice uses to
map the graph she receives to the $s$-bit message she sends.  For any
graph $H \in {\cal G}_A$, let $\Gamma(H) = \{H': \phi(H') =
\phi(H)\}$.  Since Bob needs to produce a valid vertex cover under all
inputs, Bob needs to ensure that the vertex cover output for any input
$H$ to Alice should cover every edge in the union of all graphs in
$\Gamma(H)$.  Since the number of possible outputs of $\phi$ is $2^s$,
a simple averaging argument implies that there exists an $H$ such that
$|\Gamma(H)| \ge |{\cal G}_A|/2^s$.

\begin{lemma}
  \label{lemma: union}
For any subset ${\cal F}$ of ${\cal G}$, let $I \subseteq \{1, \ldots,
k\}$ be the set of indices such that union of all of the edges in the
graphs in ${\cal F}$ has at least $(1/2 - \eps_3)n$ edges from $M_i$,
for each $i \in I$, where $\eps_3 > \eps_1$ is an arbitrary positive
constant.  If $|{\cal F}| \ge |{\cal G}_A|/2^{s(1 + o(1))}$ and
$s = o(nk)$, then $|I| = k(1 - o(1))$.
  \end{lemma}
\begin{proof}
  Fix an $i$ in $[k]$.  Let $m_i$ denote the number of edges from
  $M_i$ contained in the union of all of the edges in the graphs in
  ${\cal F}$.  Then, if $I$ is the set of indices as defined in the
  lemma, we have the following upper bound on $|{\cal F}|$.
  \begin{eqnarray*}
    |{\cal F}| & \le & \prod_i \binom{m_i}{\eps_2 n}\\
    & \le & \binom{(1/2 - \eps_1)n}{\eps_2 n}^{|I|} \binom{(1/2 - \eps_3)n}{\eps_2 n}^{k-|I|}\\
    & \le & |{\cal G}_A| \left(\frac{\binom{(1/2 - \eps_3)n}{\eps_2 n}}{\binom{(1/2 - \eps_1)n}{\eps_2 n}}\right)^{k-|I|}\\
    & \le & |{\cal G}_A| \left(\frac{1/2 - \eps_3 - \eps_2}{1/2 - \eps_1 - \eps_2}\right)^{\eps_2 n (k-|I|)}\\
    & \le & |{\cal G}_A|/2^{c n (k - |I|)},
  \end{eqnarray*}
  for some constant $c > 0$, which depends on $\eps_1, \eps_2,
  \eps_3$.  Since $|{\cal F}| \ge |{\cal G}_A|/2^{s(1 + o(1))}$, it follows that $cn(k - |I|)$ is at most $s(1 + o(1))$.  Since $s
  = o(nk)$, we obtain $|I| \ge k(1 - o(1))$.
\end{proof}

For a uniformly random chosen graph $H \in {\cal G}_A$, with
probability at least $1 - o(1)$, $|\Gamma(H)| \ge |{\cal G}_A|/2^{s(1
  + o(1))}$; this is because the number of graphs $H$ in ${\cal G}_A$
that have $|\Gamma(H)|$ less than $|{\cal G}_A|/2^{s(1 + o(1))}$ is at
most $2^s \cdot |{\cal G}_A|/2^{s(1 + o(1))}$, which is a $o(1)$
fraction of ${\cal G}_A$.  By Lemma~\ref{lemma: union}, it follows
that for a randomly chosen matching $M_r$, with probability $1 -
o(1)$, the union of all of the edges in the graphs in $\Gamma(H)$
contains at least $(1/2 - \eps_3)n$ edges from $M_r$.  Since Bob has
to cover all the edges in this union, the vertex cover returned
includes at least $n/2 - \eps_3 n$ of the vertices of $M_r$ and at
least the $2(n/2 + \eps_2 n)$ vertices needed to cover the edges of
$E_2$.  This implies a vertex cover of size at least $3n/2 + 2\eps_2 n
- \eps_3n$ with probability at least $1 - o(1)$.

This completes the proof of the desired claim that any protocol with
$o(nk) = n^{1 + \Omega(1/\lg\lg n)}$ communication (for a suitable hidden constant in the $\Omega$ term) produces a vertex cover of size at least $(3/2 +
2\eps_2 - \eps_3)n$ with probability $1 - o(1)$.  Therefore, there exists a constant $c > 0$ such that no vertex
cover computed by a two party protocol with $n^{1 + c/\lg\lg n}$ communication
complexity can have an approximation ratio of any constant smaller
than $3/2$.

\end{proof}
Our lower bound
applies to bipartite graphs. Our proof follows the approach
of~\cite{DBLP:conf/soda/GoelKK12} who established a similar lower
bound for the maximum matching problem in bipartite graphs.  While the
size of a maximum matching equals the size of a minimum vertex cover
in bipartite graphs, there is no direct reduction between the two
problems in our communication model.  Indeed, though the specific
graph construction we use for the vertex cover lower bound follows the
same framework as for the maximum matching lower bound, the parameters
and the specific arguments are different.  Our parameters for the
probability distribution over the bipartite graphs are similar to
those used in the lower bound for stochastic vertex cover
in~\cite{DBLP:conf/stoc/DerakhshanDH23}.
\section{Proof of the Main Theorem}
\label{sec:main}
In this section, we  put all the pieces together and provide a formal proof for our main result using Lemma~\ref{lemma:approximation}, Lemma~\ref{lemma:communication}, and Lemma~\ref{lemma:lower-bound}.
\begin{theorem*}[restated]
For any $k \geq 2$ and any desirably small $\epsilon > 0$, there exists a randomized MVC protocol in the $k$-party one-way communication model with an expected approximation ratio of $(2 - 2^{-k+1} + \epsilon)$, in which each party communicates a message of size \(O_{k, \epsilon}(n)\). This approximation ratio is tight for $k=2$ up to a factor of $1+\epsilon$.
\end{theorem*}
\begin{proof}
We prove in Lemma~\ref{lemma:approximation} that
if all the parties use Algorithm~\ref{algk} then the resulting protocol has an approximation ratio of  $(2 - 2^{1-k} + 5 \epsilon')$ for  any $\epsilon' \in (0, 1/5)$. Moreover, based on Lemma~\ref{lemma:communication}, the communication cost of this protocol is $(k-1)! \log_{1 + \epsilon'}2^{k-1} O(n)$ which is linear in $n$ assuming $k$ and $\epsilon'$ are constants. Therefore, the communication cost of our protocol is $O_{k,\epsilon'}(n)$. By choosing a sufficiently small $\epsilon'\leq \epsilon/5$ we get an approximation ratio of $(2 - 2^{1-k} + \epsilon)$ and communication cost of $O_{\epsilon, k}(n)$.
    
Finally, to prove the tightness of our approximation ratio for the two-party  case,  in Lemma~\ref{lemma:lower-bound}, we have a lower bound of $3/2 - \epsilon$ for the expected approximation ratio of two-party protocols with communication cost of $O(n)$. This implies that our protocol for two parties is tight within an additive error of  $\epsilon$.
\end{proof}

\bibliographystyle{plainnat}
\bibliography{ref}

\end{document}